%% file: sn-article.tex
\documentclass[pdflatex,sn-mathphys-num]{sn-jnl}

\usepackage{graphicx}%
\usepackage{multirow}%
\usepackage{amsmath,amssymb,amsfonts}%
\usepackage{amsthm}%
\usepackage{mathrsfs}%
\usepackage[title]{appendix}%
\usepackage{xcolor}%
\usepackage{textcomp}%
\usepackage{manyfoot}%
\usepackage{booktabs}%
\usepackage{algorithm}%
\usepackage{algorithmicx}%
\usepackage{algpseudocode}%
\usepackage{listings}%

\theoremstyle{thmstyleone}%
\newtheorem{theorem}{Theorem}
\newtheorem{proposition}[theorem]{Proposition}%

\theoremstyle{thmstyletwo}%

\theoremstyle{thmstylethree}%
\newtheorem{definition}{Definition}%

\begin{document}

\title[Article Title]{Contextuality as a Left Adjoint:
A Categorical Generation of Orthomodular Structure}


\author*[1]{\fnm{Yukio-Pegio} \sur{Gunji}}\email{pegioyukio@gmail.com}

\author[1]{\fnm{Yoshihiko} \sur{Ohzawa}}\email{yoshi1056@icloud.com}

\author[1]{\fnm{Yuki} \sur{Tokuyama}}\email{yukiyuuki205@gmail.com}

\author[1]{\fnm{Yu} \sur{Huang}}\email{kouiku921@gmail.com }

\author[2]{\fnm{Kyoko} \sur{Nakamura}}\email{kyoko608@gmail.com }

\affil*[1]{\orgdiv{Intermedia Art and Science, School of Fundamental Science and Technology}, \orgname{Waseda University}, \orgaddress{\street{Ohkubo 3-4-1}, \city{Shinjuku-ku}, \postcode{169-8555}, \state{Tokyo}, \country{Japan}}}

\affil[2]{\orgdiv{Department of Fine and Applied Arts}, \orgname{Kyoto University of the Art}, \orgaddress{\street{2-116 Uryumyama, Kitashirakawa, Sakyo-ku}, \city{Kyoto}, \postcode{606-8271}, \state{Kyoto}, \country{Japan}}}


\abstract
{Contextuality is widely regarded as a hallmark of quantum information,
yet its structural origin is often obscured by probabilistic or operational
formulations.
In this work, we show that non-distributive orthomodular structure need not
be postulated, but arises canonically as a \emph{left adjoint} from classical
Boolean contexts.

We introduce a gluing functor that takes pairs of Boolean algebras and
identifies only their minimal and maximal elements via a categorical pushout.
The resulting lattice is orthomodular but generically non-distributive.
We prove that this construction is left adjoint to a forgetful functor
selecting Boolean subalgebras, thereby providing a free but constrained
generation of quantum-logical structure from classical contexts.

Furthermore, we demonstrate that the failure of this pushout to remain Boolean
is equivalent to the absence of global sections in the sheaf-theoretic framework
of Abramsky and Brandenburger.
This establishes a precise correspondence between contextuality as a sheaf
obstruction and non-distributivity as a colimit failure.

Our results offer a categorical and lattice-theoretic reconstruction of
contextuality that precedes probabilistic notions and clarifies the structural
necessity of quantum logic in information-theoretic settings.}




\maketitle

\section{Introduction}

Contextuality is widely regarded as a defining feature of quantum theory and a key
resource for quantum information processing.
It underlies the impossibility of noncontextual hidden-variable models in the sense of
Kochen--Specker, the failure of a single global classical description compatible with
all measurement settings, and the emergence of quantum advantages in computation and
communication \cite{KochenSpecker1967,Spekkens2005,CabelloSeveriniWinter2014,HowardWallmanVeitchEmerson2014,Raussendorf2013,BudroniEtAl2022}.
In the modern formulation, contextuality has acquired a precise structural meaning:
a family of locally classical descriptions (``contexts'') may fail to admit any single
global classical model \cite{AbramskyBrandenburger2011,Abramsky2015}.

Concrete demonstrations of contextuality were subsequently given by a variety
of no-go theorems and inequalities.
Beyond the original Kochen--Specker construction, state-dependent
inequalities such as the KCBS inequality \cite{Klyachko2008}
and operational formulations due to Peres and Mermin
\cite{Peres1991,Peres1995,Mermin1993}
made clear that contextuality is not a marginal logical curiosity,
but an experimentally accessible and structurally robust feature of
quantum theory.
These developments further strengthened the view that non-Boolean
logical structure is indispensable for an adequate description of
quantum phenomena.

From the perspective of logic and lattice theory, quantum phenomena have long been
associated with the breakdown of Boolean distributivity.
Birkhoff and von Neumann proposed that the propositional structure of quantum theory
should be described by a non-distributive lattice, with the lattice of closed subspaces
of a Hilbert space as the paradigmatic example \cite{BirkhoffvonNeumann1936}.
This initiated the program of quantum logic, developed further through the study of
orthomodular lattices and related axiomatics \cite{Piron1976,Kalmbach1983,MaedaMaeda1970,DvurecenskijPulmannova2000}.
In this tradition, orthomodularity and non-distributivity typically enter as
\emph{axiomatic} characteristics of the quantum propositional calculus.

In parallel, sheaf- and category-theoretic approaches clarified contextuality as an
obstruction to global consistency.
The Abramsky--Brandenburger framework represents measurement scenarios by a cover
of contexts and encodes empirical models as compatible families of local data; contextuality
appears precisely as the absence of global sections \cite{AbramskyBrandenburger2011,AbramskyMansfieldBarbosa2012}.
Related categorical reconstructions of quantum theory emphasize compositionality and
process structure, often taking nonclassical structure as part of the ambient framework
\cite{AbramskyCoecke2004,AbramskyCoecke2008,HeunenLandsmanSpitters2009}.
More broadly, contextuality has been studied in generalized probabilistic and
hypergraph-like settings, connecting classical representability with combinatorial and
geometric constraints \cite{FoulisRandall1981,Pitowsky1989}.

Despite these advances, a conceptual gap remains.
On the one hand, quantum logic highlights that non-distributive orthomodular structure
is deeply tied to quantum phenomena \cite{BirkhoffvonNeumann1936,Kalmbach1983}.
On the other hand, the sheaf-theoretic approach identifies contextuality as a
\emph{consistency obstruction} without explaining, at the level of logical structure,
\emph{why} the minimal nonclassical extension should be orthomodular and non-distributive
rather than some other non-Boolean form \cite{AbramskyBrandenburger2011,Abramsky2015,BudroniEtAl2022}.
In most formulations, orthomodularity is either assumed from the outset (quantum logic)
or inherited from the Hilbert-space formalism (quantum probability), rather than derived.

The present work addresses this gap by adopting a generative, categorical viewpoint.
Our starting point is deliberately classical: we assume only Boolean algebras representing
local classical logics (contexts) and their homomorphisms.
We then ask a structural question:
\emph{What is the minimal non-Boolean propositional structure forced upon us by composing
multiple Boolean contexts under shared extremal elements?}
This question is natural not only in foundational quantum theory but also in models of
decision making and cognition where classical descriptions are context-dependent
\cite{Sozzo2021,Pivato2025,Hammond2025a,Chichilnisky2022,Piermont2023}.

Our main construction is a \emph{gluing functor} defined via a categorical pushout.
Given two Boolean algebras $B_1$ and $B_2$ sharing only the bottom and top elements,
we form the pushout along the two-element Boolean algebra $\mathbf{2}=\{0<1\}$.
Intuitively, this operation identifies only $0$ and $1$ while keeping all other
propositional content distinct, producing the minimal composite structure consistent
with the two local logics.
We prove that the resulting lattice is generically non-distributive and admits a
canonical orthocomplementation making it an orthomodular lattice.
Hence, quantum-logical structure is \emph{not postulated} but \emph{generated} from
classical contexts.

Categorically, we show that this gluing construction is left adjoint to a forgetful
functor that extracts the chosen Boolean subalgebras (contexts) from a context-labelled
orthomodular lattice.
This adjunction formalizes the sense in which our construction is the \emph{free but
constrained} extension of classical logic compatible with multiple contexts.
Finally, we establish a precise correspondence with the Abramsky--Brandenburger picture:
the failure of the pushout to remain Boolean is equivalent to the absence of global
sections in the associated presheaf model.
In this way, contextuality is identified with a \emph{colimit failure} in the category
of Boolean algebras, while non-distributive orthomodular structure appears as the minimal
algebraic realization of this obstruction.

\paragraph{Contributions.}
The contributions of this paper are:
(i) a pushout-based functorial construction that generates orthomodular, non-distributive
structure from Boolean contexts;
(ii) an explicit adjunction ($\mathcal{G}\dashv \mathcal{U}$) establishing the construction
as a left-adjoint ``quantization'' of classical contexts;
(iii) a structural equivalence between sheaf-theoretic contextuality and the non-Boolean
character of the Boolean pushout.
These results clarify the logical origin of contextuality prior to probabilistic and
Hilbert-space assumptions and provide a bridge between quantum logic and sheaf-theoretic
foundations.

\section{Results}

In this section we make explicit the functorial structure underlying the
gluing construction and establish its universal and algebraic consequences.
All categorical preliminaries are deferred to the Methods.

Figure~\ref{fig:Fig-1} illustrates the central construction of this work.
While the coproduct of Boolean algebras freely combines contexts without
constraints, the pushout along the two-element Boolean algebra imposes the
minimal identification required for logical coherence.
This operation leaves each Boolean context intact but forces the global
structure to become non-distributive.
The resulting lattice is orthomodular and thus realizes quantum logic as
a generated, rather than assumed, structure. Elements \(a, b\), and \(c\) show the breakdown of distributive law in the glued lattice.

\begin{figure}[h]
    \centering
    \includegraphics[width=1\linewidth]{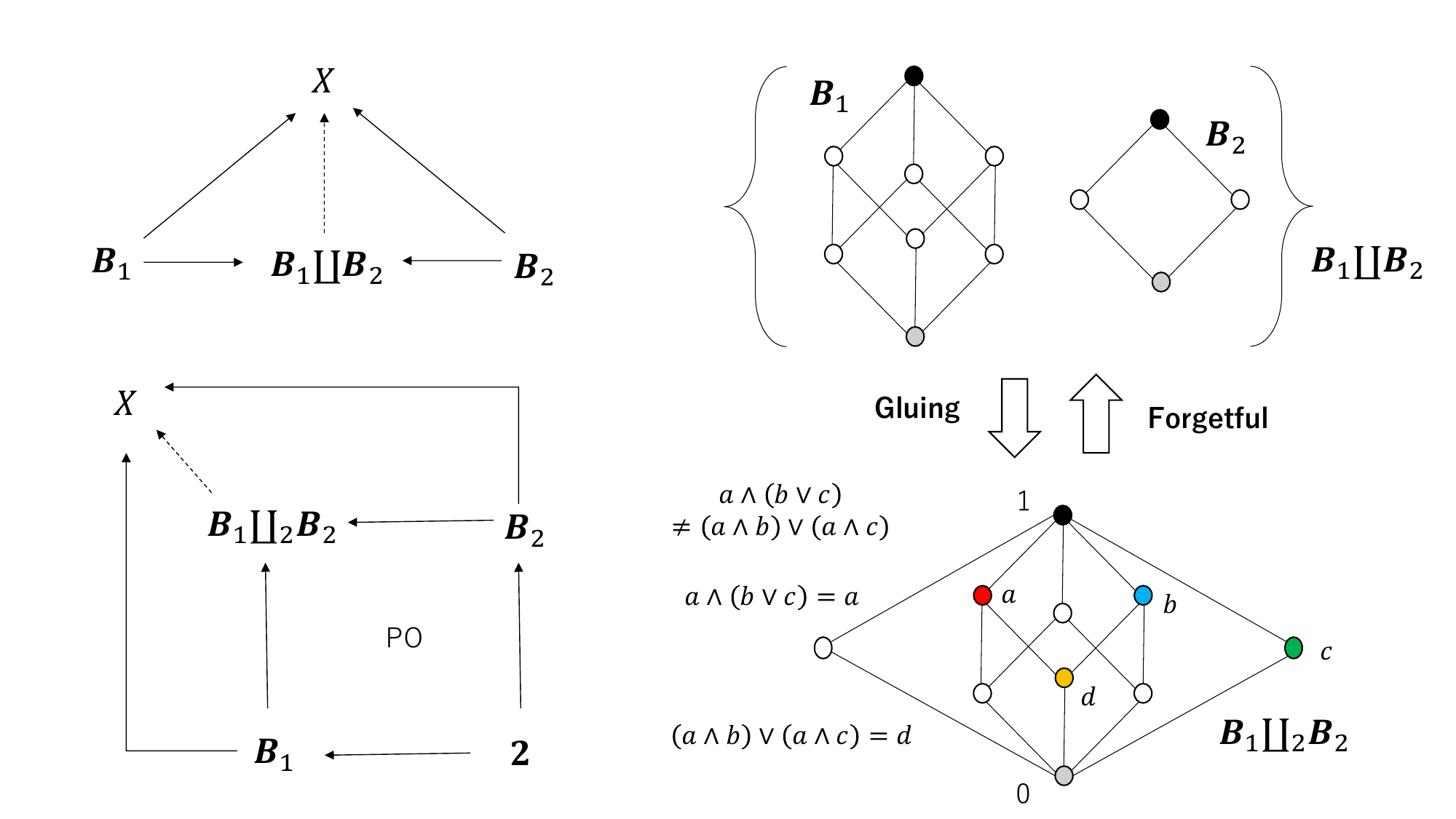}
    \caption{From coproduct to quantum logic via pushout.
(Top left) The coproduct $B_1 + B_2$ of Boolean algebras, characterized by
the universal property without identifications.
(Bottom left) The pushout $B_1 \amalg_{\mathbf{2}}  B_2$, obtained by imposing
the identification of bottom and top elements through morphisms from the
two-element Boolean algebra $\mathbf{2}$.
(Top right) Hasse diagrams of a $2^3$-Boolean algebra and a $2^2$-Boolean
algebra representing two classical contexts.
(Bottom right) The resulting non-distributive orthomodular lattice obtained
by gluing the two Boolean algebras along their extremal elements.
The gluing functor $\mathcal{G}$ produces a quantum-logical structure while
preserving each Boolean block as a context.}
    \label{fig:Fig-1}
\end{figure}

\subsection{The forgetful functor}

We first define the forgetful functor appearing in the adjunction.

\begin{definition}
Let $\mathbf{OMLat}_{\mathrm{ctx}}^{(2)}$ denote the category whose
objects are triples $(L,B_1,B_2)$ such that
\begin{itemize}
\item $L$ is an orthomodular lattice,
\item $B_1,B_2 \subseteq L$ are Boolean subalgebras,
\item $B_1$ and $B_2$ share the same bottom and top elements.
\end{itemize}
Morphisms $(L,B_1,B_2)\to(L',B'_1,B'_2)$ are orthomodular lattice
homomorphisms $h:L\to L'$ satisfying
\begin{equation}
    h(B_1)\subseteq B'_1,\qquad h(B_2)\subseteq B'_2.
\end{equation}

\end{definition}

\begin{definition}[Forgetful functor]
Define a functor
\begin{equation}
\mathcal{U}:\mathbf{OMLat}_{\mathrm{ctx}}^{(2)}
\longrightarrow
\mathbf{Bool}_{0,1}\times\mathbf{Bool}_{0,1}
\end{equation}
by
\begin{equation}
    \mathcal{U}(L,B_1,B_2) := (B_1,B_2),
\end{equation}

and on morphisms
\begin{equation}
\mathcal{U}(h) := (h|_{B_1},\,h|_{B_2}).
\end{equation}
\end{definition}

\subsection{Adjunction and natural isomorphism}

Let
\begin{equation}
\mathcal{G}:\mathbf{Bool}_{0,1}\times\mathbf{Bool}_{0,1}
\to
\mathbf{OMLat}_{\mathrm{ctx}}^{(2)}
\end{equation}
be the gluing functor defined in the Methods.

\begin{theorem}[Adjunction]
The functors $\mathcal{G}$ and $\mathcal{U}$ form an adjoint pair
\begin{equation}
\mathcal{G} \dashv \mathcal{U}.
\end{equation}
\end{theorem}

\begin{proof}
For any $(B_1,B_2)\in\mathbf{Bool}_{0,1}^2$ and any
$(L,C_1,C_2)\in\mathbf{OMLat}_{\mathrm{ctx}}^{(2)}$, define

\begin{equation}
\begin{aligned}
\Phi_{(B_1,B_2),(L,C_1,C_2)} \; :\;
& \mathrm{Hom}_{\mathbf{OMLat}_{\mathrm{ctx}}^{(2)}}
\bigl(\mathcal{G}(B_1,B_2),(L,C_1,C_2)\bigr) \\
& \longrightarrow
\mathrm{Hom}_{\mathbf{Bool}_{0,1}}(B_1,C_1)
\times
\mathrm{Hom}_{\mathbf{Bool}_{0,1}}(B_2,C_2).
\end{aligned}
\end{equation}

by restriction to the Boolean blocks.

Conversely, given Boolean homomorphisms
$f_1:B_1\to C_1$ and $f_2:B_2\to C_2$ agreeing on $0$ and $1$,
the universal property of the pushout yields a unique orthomodular
homomorphism
\begin{equation}
\langle f_1,f_2\rangle:\mathcal{G}(B_1,B_2)\to L.
\end{equation}
These assignments are mutually inverse, establishing a bijection.
\end{proof}

\begin{proposition}
The family of bijections
\begin{equation}
\Phi_{(B_1,B_2),(L,C_1,C_2)}
\end{equation}
is natural in both arguments. Hence the adjunction
$\mathcal{G}\dashv\mathcal{U}$ is witnessed by a natural isomorphism
of Hom-functors.
\end{proposition}

\begin{proof}
Naturality follows from functoriality of restriction and from the
uniqueness part of the pushout universal property.
For any morphisms on either side, the corresponding squares of Hom-sets
commute by construction.
\end{proof}

\subsection{Orthomodularity of the glued lattice}

\begin{proposition}
Let $P = B_1 \amalg_{\mathbf{2}} B_2$ be the pushout \cite{Awodey2010,MacLane1998} constructed in the Methods.
Then $P$ admits a canonical orthocomplementation and is an orthomodular lattice.
\end{proposition}

\begin{proof}
Each Boolean block $B_i$ carries a Boolean complement.
Since the blocks intersect only at $0$ and $1$, these complements agree
on the intersection and therefore define a global orthocomplementation.
The orthomodular law holds within each block and hence globally.
\end{proof}

\subsection{Failure of distributivity}

\begin{proposition}
If $B_1$ and $B_2$ are nontrivial and distinct Boolean algebras, then
the lattice $P=B_1\amalg_{\mathbf{2}}B_2$ is non-distributive.
\end{proposition}

\begin{proof}
Choose $a\in B_1\setminus\{0,1\}$ and $b,c\in B_2\setminus\{0,1\}$ with $b\neq c$.
Since $a,b,c$ do not belong to a common Boolean block, the distributive
identity fails in general. Hence $P$ cannot be Boolean.
\end{proof}

\begin{figure}[h]
    \centering
    \includegraphics[width=1\linewidth]{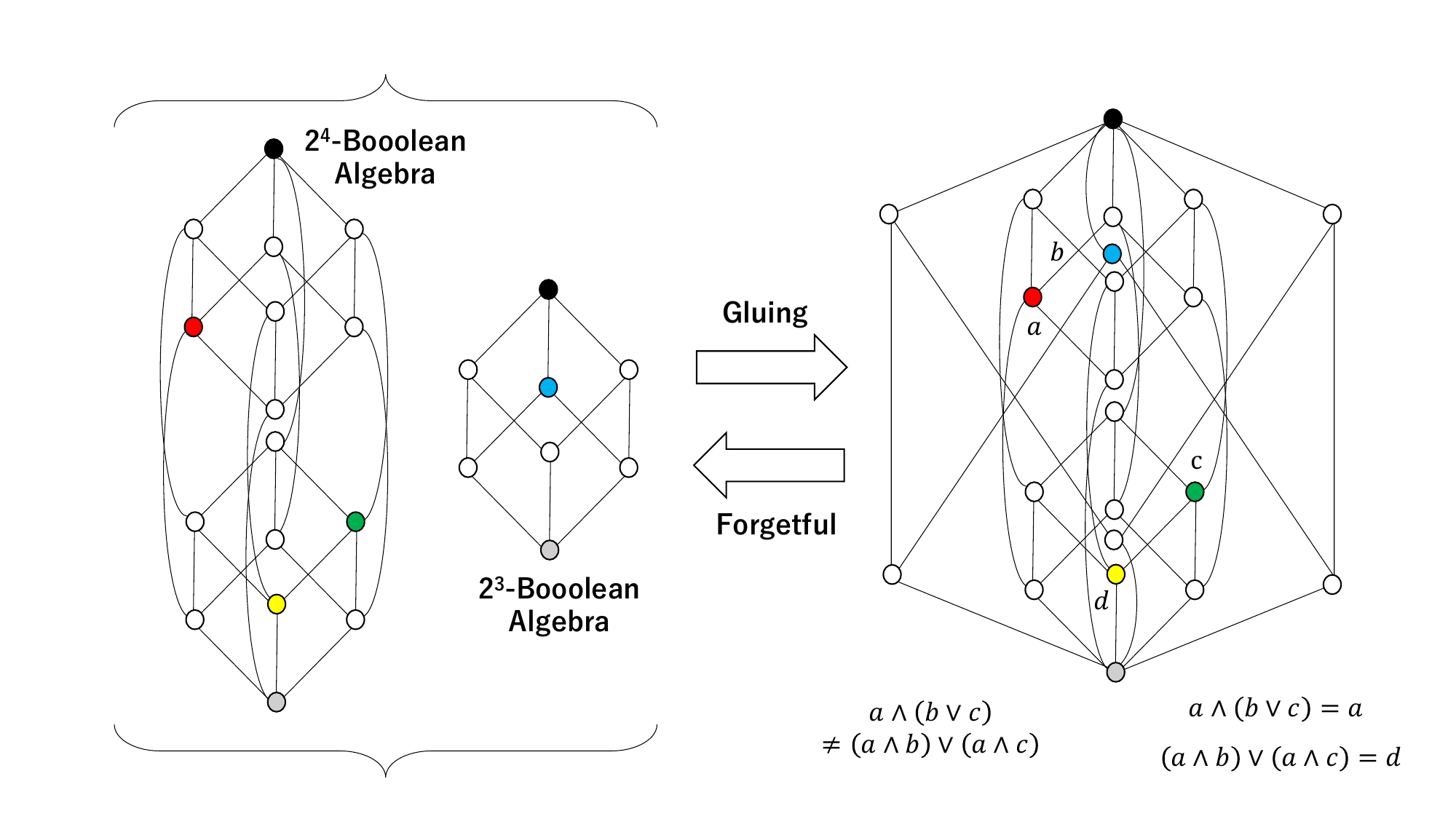}
    \caption{Gluing Boolean contexts of different sizes.
(Top) Hasse diagrams of a $2^3$-Boolean algebra and a $2^2$-Boolean algebra.
(Bottom) The orthomodular lattice obtained by gluing these Boolean algebras
along their shared extremal elements.
The failure of distributivity persists despite the asymmetry in context size,
demonstrating that orthomodularity is the minimal consistent extension
independent of the relative complexity of the Boolean blocks.}
    \label{fig:Fig-2}
\end{figure}

Figure~\ref{fig:Fig-2} emphasizes that the emergence of quantum-logical structure does not
depend on symmetry or equal dimensionality of the classical contexts.
Even when Boolean algebras of different sizes are glued, distributivity is
necessarily broken, while orthomodularity is preserved.
This robustness highlights that non-distributive structure is a structural
consequence of contextual composition itself, rather than a feature of
particular Boolean realizations.

\subsection{Morphisms between glued lattices}

\begin{proposition}
Let $P=\mathcal{G}(B_1,B_2)$ and $P'=\mathcal{G}(B'_1,B'_2)$ be glued
non-distributive orthomodular lattices.
Any pair of Boolean homomorphisms
\begin{equation}
f_1:B_1\to B'_1,\qquad f_2:B_2\to B'_2
\end{equation}
that agree on $0$ and $1$ uniquely extends to an orthomodular lattice
homomorphism
\begin{equation}
\tilde f:P\to P'.
\end{equation}
\end{proposition}

\begin{proof}
This follows directly from the functoriality of the pushout construction.
The extension $\tilde f$ is induced by the universal property and preserves
orthocomplementation blockwise.
\end{proof}

\subsection{Relation to Free Orthomodular Lattices}

Let $X_1,X_2$ be sets and let
\begin{equation}
B_i \cong \mathcal{P}(X_i)
\end{equation}
be free Boolean algebras.

\begin{theorem}
The lattice $\mathcal{G}(B_1,B_2)$ is the free orthomodular lattice
generated by two Boolean contexts
subject only to the identification of $0$ and $1$.
\end{theorem}

In other words,
$\mathcal{G}(B_1,B_2)$ is the free orthomodular lattice on the partial
Boolean algebra obtained by gluing $B_1$ and $B_2$ at $\{0,1\}$.

\subsubsection{Connection to the Abramsky--Brandenburger framework}

We now clarify the precise relationship between our lattice-theoretic
construction and the sheaf-theoretic formulation of contextuality
introduced by Abramsky and Brandenburger.
While both approaches address contextuality, they do so at different
structural levels.

In the Abramsky--Brandenburger framework, contextuality is identified as
the absence of global sections of a presheaf of measurement outcomes.
Our construction provides a complementary algebraic perspective,
showing how this obstruction arises from the failure of a Boolean
colimit to remain Boolean.

\subsubsection{Dictionary}

\begin{table}[h!]
\centering
\caption{Comparison of Abramsky-Brandenburger and preset approach}
\begin{tabular}{@{}ccc@{}}
Abramsky--Brandenburger & Present framework \\ \hline
Context & Boolean algebra $B_i$ \\
Overlap of contexts & $\mathbf{2}=\{0,1\}$ \\
Presheaf & Disjoint family of $B_i$ \\
Sheaf condition & Global Boolean algebra \\
Contextuality & Non-distributivity of $\mathcal{G}$ \\
\end{tabular}
\end{table}

Table 1  should be read not merely as a translation between
terminologies, but as identifying the same structural obstruction
expressed in two different mathematical languages.

\subsection{Sheaf collapse and colimit failure}

In the sheaf-theoretic setting\cite{AbramskyBrandenburger2011,AbramskyMansfieldBarbosa2012,Abramsky2015}, the existence of a global section
corresponds to the possibility of consistently assigning outcomes
across all measurement contexts.
The absence of such a section signals contextuality.

From the perspective developed in this work, the same obstruction
appears as the failure of the pushout
\begin{equation}
B_1 \amalg_{\mathbf{2}} B_2
\end{equation}
to remain a Boolean algebra.
If a global Boolean algebra embedding both contexts existed,
the pushout would be Boolean; its non-distributivity therefore
captures precisely the impossibility of a global section.

Thus, contextuality in the sense of Abramsky and Brandenburger is
equivalent to the emergence of non-distributive orthomodular structure
under categorical gluing.
The sheaf-theoretic obstruction and the lattice-theoretic colimit
failure represent the same phenomenon at different structural levels.

\section{Discussion}

\begin{figure} [h]
    \centering
    \includegraphics[width=1\linewidth]{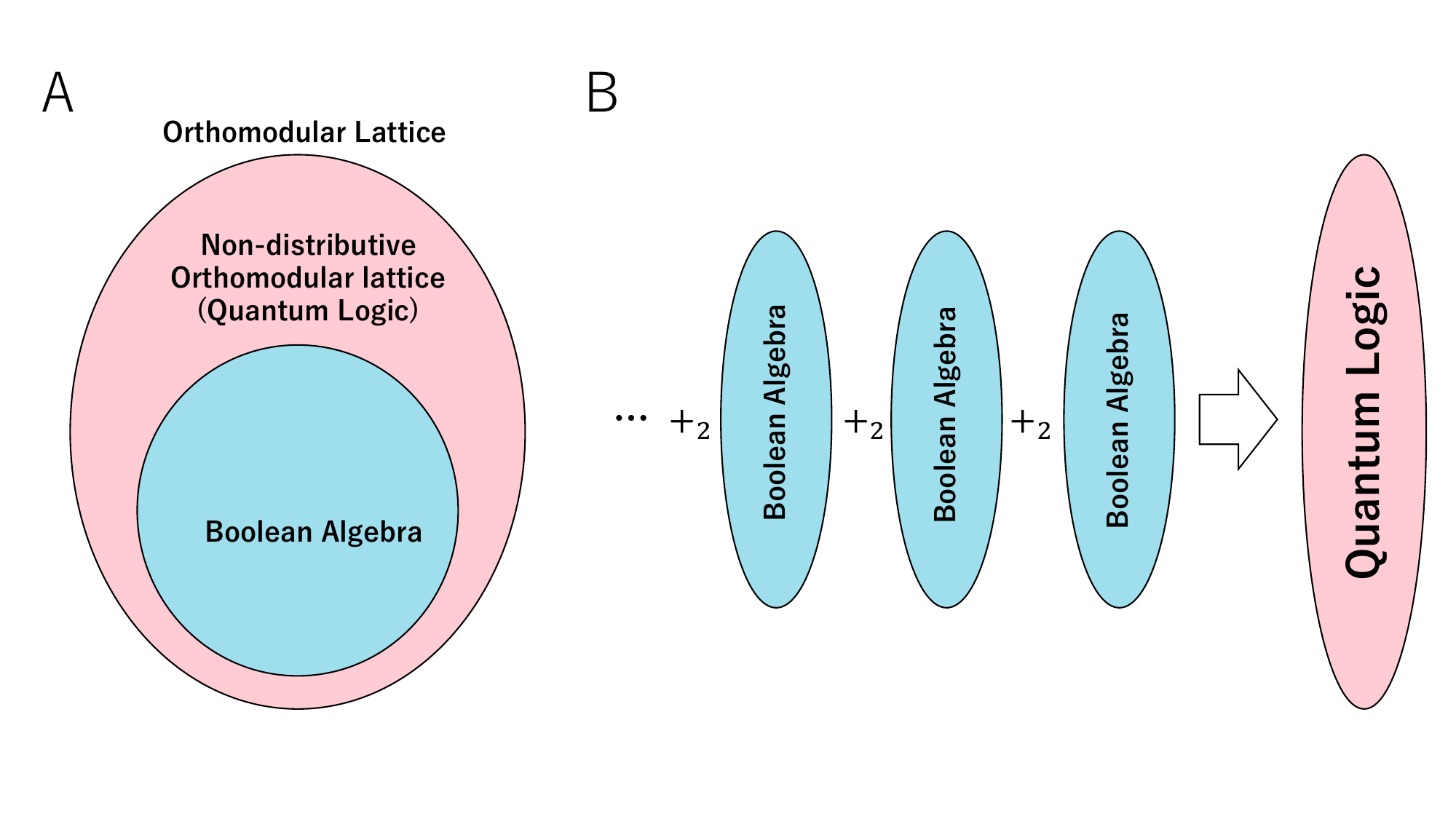}
    \caption{Contrasting views on the origin of quantum logic.
(A) The traditional picture, in which an orthomodular lattice is taken as
the ambient structure, and Boolean algebras appear as embedded classical
substructures.
(B) The constructive picture proposed in this work, where multiple Boolean
algebras are glued via the operation $+_{\mathbf{2}}$, generating quantum
logic as a non-distributive orthomodular lattice.
In this perspective, quantum logic is not assumed but arises as the minimal
structure compatible with multiple classical contexts.}
    \label{fig:Fig-3}
\end{figure}

Figure~\ref{fig:Fig-3} summarizes the conceptual shift introduced by the present work.
In the standard view (Fig.~3A), orthomodular structure is postulated from the
outset, and Boolean algebras are recovered as classical fragments.
By contrast, our construction (Fig.~3B) reverses this logical direction:
starting exclusively from Boolean contexts, quantum logic emerges through
their categorical gluing.
This inversion clarifies that non-distributive orthomodular structure should
be understood as an outcome of contextual composition, rather than as a
primitive postulate of quantum theory.

The central conceptual message of this work is that non-distributive
orthomodular structure does not need to be postulated as a primitive
feature of quantum theory.
Instead, it arises inevitably from the categorical composition of
classical Boolean contexts.

In contrast to traditional approaches to quantum logic, which assume
orthomodularity\cite{BirkhoffvonNeumann1936,Kalmbach1983,Piron1976} or non-distributivity at the outset, our construction
starts exclusively from Boolean algebras.
The emergence of quantum logic is therefore a derived, rather than
axiomatic, phenomenon.

Notably, no reference is made to Hilbert spaces, quantum states,
or probabilistic structures.
The only ingredients are classical logical contexts and their
minimal categorical gluing.

This perspective suggests a shift in how quantum structure is interpreted.
Rather than attributing non-Boolean logic to microscopic systems themselves,
our results indicate that quantum logic reflects the necessity of reconciling
multiple, mutually incompatible classical descriptions.

The present construction is compatible with experimentally observed
forms of contextuality, including Kochen--Specker- and KCBS-type
inequalities \cite{KochenSpecker1967,Klyachko2008,Peres1995,Mermin1993},
while shifting the emphasis from probabilistic violations to the
logical and categorical origin of nonclassical structure.

Within the sheaf-theoretic framework of Abramsky and Brandenburger,
contextuality is identified as the absence of global sections.
Our results sharpen this picture by showing that such an obstruction
is precisely equivalent to the failure of a Boolean pushout to remain Boolean.
In this sense, non-distributivity provides the structural origin of
contextuality, rather than merely its symptom.

From a foundational perspective, our approach is complementary to
topos-theoretic reconstructions of quantum logic
\cite{DoringIsham2008,DoringIsham2011},
which likewise emphasize contextuality but retain a richer internal
logical structure.
In contrast, our results isolate orthomodular, non-distributive logic
as the minimal algebraic consequence of gluing Boolean contexts.

To our knowledge, this is the first formulation in which quantum logic
is obtained functorially as a left adjoint from classical Boolean contexts,
rather than assumed as an independent axiom.

From this viewpoint, quantum theory appears not as a departure from
classical logic, but as the minimal logical extension forced upon us
by contextual consistency\cite {HowardWallmanVeitchEmerson2014,CabelloSeveriniWinter2014,Raussendorf2013}.

\section{Conclusion}

In this work, we have shown that non-distributive orthomodular structure
need not be assumed as a primitive feature of quantum theory.
Instead, it arises canonically from the categorical composition of
classical Boolean contexts.

Starting solely from Boolean algebras and their minimal gluing along
shared extremal elements, we constructed a non-distributive orthomodular
lattice via a pushout.
We demonstrated that this construction defines a gluing functor which
is left adjoint to a natural forgetful functor extracting Boolean
subalgebras.
As a consequence, quantum-logical structure appears as a free but
constrained extension of classical logic, rather than as an independent
axiom.

We further established that the failure of this pushout to remain Boolean
is exactly equivalent to contextuality in the sense of
Abramsky and Brandenburger.
This provides a precise algebraic characterization of contextuality as a
colimit failure in the category of Boolean algebras, thereby identifying
non-distributivity as its structural origin.

Taken together, our results suggest a reinterpretation of quantum logic:
it should not be viewed as a property intrinsic to microscopic systems,
but as an unavoidable consequence of reconciling multiple, incompatible
classical descriptions within a single coherent framework.
In this sense, quantum logic precedes probabilistic and Hilbert-space
formulations, and reflects a deeper categorical necessity.

By deriving orthomodular, non-distributive structure functorially from
Boolean contexts, this work clarifies the logical foundations of
contextuality and opens new perspectives on quantum-like information
processing beyond standard quantum formalism.

\section{Methods}

\subsection{The category of Boolean algebras}

We begin by fixing the categorical setting.

\begin{definition}
Let $\mathbf{Bool}$ be the category defined as follows:
\begin{itemize}
\item Objects are Boolean algebras of the form $\mathcal{P}(X)$,
the power set of a set $X$, equipped with the usual operations
$\cap$, $\cup$, complement ${}^{c}$, and distinguished elements
$\emptyset$ and $X$.
\item Morphisms $h:\mathcal{P}(X)\to\mathcal{P}(Y)$ are bounded lattice
homomorphisms, i.e.\ maps preserving $\cap$, $\cup$, $\emptyset$, and $X$.
\end{itemize}
\end{definition}

\begin{proposition}
$\mathbf{Bool}$ is a category.
\end{proposition}

\begin{proof}
Let $h:\mathcal{P}(X)\to\mathcal{P}(Y)$ and
$k:\mathcal{P}(Y)\to\mathcal{P}(Z)$ be Boolean homomorphisms.
Then their composition $k\circ h$ preserves $\cap$, $\cup$,
$\emptyset$, and $Z$, hence is again a Boolean homomorphism.
Thus morphisms are closed under composition.

For each object $\mathcal{P}(X)$, the identity map
$\mathrm{id}_{\mathcal{P}(X)}$ preserves all Boolean operations and
serves as an identity morphism.

Associativity of composition and the identity laws follow from
the corresponding properties of function composition.
\end{proof}

\subsection{Pushouts in a category}

We recall the definition of a pushout.

\begin{definition}
Given a span
\begin{equation}
A \xrightarrow{f} B_1, \qquad A \xrightarrow{g} B_2
\end{equation}
in a category $\mathcal{C}$, a \emph{pushout} consists of an object $P$
together with morphisms
\begin{equation}
i_1:B_1\to P, \qquad i_2:B_2\to P
\end{equation}
such that:
\begin{enumerate}
\item (Commutativity)
\begin{equation}
i_1\circ f = i_2\circ g;
\end{equation}
\item (Universality)
for any object $X$ and morphisms
$h_1:B_1\to X$, $h_2:B_2\to X$ satisfying
$h_1\circ f = h_2\circ g$, there exists a unique morphism
$u:P\to X$ such that
\begin{equation}
u\circ i_1 = h_1, \qquad u\circ i_2 = h_2.
\end{equation}
\end{enumerate}
\end{definition}

\subsection{Existence of pushouts in $\mathbf{Bool}$}

Let $\mathbf{2}=\{\emptyset,\{*\}\}$ be the two-element Boolean algebra.
Let $B_1=\mathcal{P}(X)$ and $B_2=\mathcal{P}(Y)$ be objects of
$\mathbf{Bool}$.
Define Boolean homomorphisms
\begin{equation}
f:\mathbf{2}\to B_1,\qquad g:\mathbf{2}\to B_2
\end{equation}
by
\begin{equation}
f(\emptyset)=\emptyset,\quad f(\{*\})=X,
\qquad
g(\emptyset)=\emptyset,\quad g(\{*\})=Y.
\end{equation}

We construct the pushout explicitly.
Let $P$ be the Boolean algebra generated by the disjoint union
$X \sqcup Y$ subject only to the identifications
\begin{equation}
\emptyset_{B_1} = \emptyset_{B_2}, \qquad
X = Y,
\end{equation}
that is, $P$ is the quotient of $\mathcal{P}(X)\sqcup\mathcal{P}(Y)$
by identifying the bottom and top elements.

Define morphisms
\begin{equation}
i_1:B_1\to P,\qquad i_2:B_2\to P
\end{equation}
as the canonical injections.
Then
\begin{equation}
i_1\circ f = i_2\circ g
\end{equation}
holds by construction, since both compositions map $\emptyset$ to the
bottom of $P$ and the top element to the top of $P$.

\begin{proposition}
The object $P$ together with $i_1$ and $i_2$ is the pushout of
$f$ and $g$ in $\mathbf{Bool}$.
\end{proposition}

\begin{proof}
Let $X'$ be any Boolean algebra and let
$h_1:B_1\to X'$, $h_2:B_2\to X'$ be Boolean homomorphisms such that
\begin{equation}
h_1\circ f = h_2\circ g.
\end{equation}
This condition implies
\begin{equation}
h_1(\emptyset)=h_2(\emptyset), \qquad
h_1(X)=h_2(Y).
\end{equation}

Define $u:P\to X'$ by
\begin{equation}
u(i_1(b)) := h_1(b), \qquad
u(i_2(c)) := h_2(c),
\end{equation}
for $b\in B_1$, $c\in B_2$.
This assignment is well-defined because the only identifications in $P$
are of the bottom and top elements, which are respected by $h_1$ and $h_2$.
It is straightforward to verify that $u$ is a Boolean homomorphism.

Uniqueness follows from the fact that every element of $P$ is represented
by an element of either $B_1$ or $B_2$.
\end{proof}

Thus, the category $\mathbf{Bool}$ admits pushouts along the extremal
elements, and the resulting construction identifies only the minimal and
maximal elements, imposing no further relations.

\section*{Declarations}

\begin{itemize}
\item Conflict of interest/Competing interests: The authors declare no conflict of interest.
\item Funding: This work is supported by institutional research funds of Waseda University.
\item Ethics approval and consent to participate: Not applicable
\item Data availability: Not applicable.
\item Materials availability: Not applicable.
\item Code availability: Not applicable.
\item Clinical trial number: Not applicable.
\item Author contribution: YPG and KN made a grand design of this approach, and YPG provided basic mathematical definitions. YPG, YO, YT, YH discussed and examined category-theoretic framework. YPG mainly wrote the manuscript, and YO, YT, YH and KN partly wrote and revised the manuscript. All members read completed one.
\end{itemize}


\input{outputs.bbl}


\end{document}